\documentclass[runningheads]{llncs}

\usepackage{amsmath,amssymb,url}
\usepackage{booktabs}
\usepackage{multirow}
\usepackage{tikz}
\usepackage{pgfplots}
\usepackage{graphicx}
\usepackage[
  n,
  operators,
  advantage, 
  sets,
  adversary,
  landau,
  probability,
  notions,  
  logic, 
  ff,
  mm,
  primitives,
  events,
  complexity,
  asymptotics,
  keys]{cryptocode} 
  \newcommand\omicron{O}
\begin{document}
\title{Supersingular Isogeny Oblivious Transfer (SIOT)}
%
%
\author{Paulo Barreto\inst{1} \and
Anderson Nascimento\inst{1}\and
Gl\'{a}ucio Oliveira\inst{2}\and Waldyr Benits\inst{3}}
\authorrunning{Barreto, P., Nascimento, A., et al.}
%
\institute{University of Washington, Tacoma, USA \\
\email{\{pbarreto,andclay\}@uw.edu} \\
\url{http://directory.tacoma.uw.edu/employee/pbarreto}, \url{http://directory.tacoma.uw.edu/employee/andclay} \and
Institute of Mathematics and Statistics, University of S\~{a}o Paulo, Brazil
\email{glaucioaorj@gmail.com} \and
Naval Systems Analysis Center, Brazilian Navy, Brazil\\
\email{wbenits@yahoo.com.br}}
\maketitle              
\begin{abstract}
In this paper we present an \emph{Oblivious Transfer} (OT) protocol that combines an OT
  scheme together with the
  \emph{Supersingular Isogeny Diffie-Hellman} (SIDH) primitive. Our proposal is a candidate for
  \emph{post-quantum} secure OT and demonstrates that SIDH naturally supports
  OT functionality. We consider the protocol in the simplest
  configuration of $\binom{2}{1}$-SIOT and analyze the 
  protocol to verify its security. 

\end{abstract}
\section{Introduction}

The first notion of Oblivious Transfer (OT) was proposed in~\cite{rabin:1981}. Most, if not all, of cryptography  protocols can be based on the notion of OT, under the assumption that an efficient
OT scheme is available.

Efficient OT protocols are known in a
\textit{quantum}-susceptible scenario~\cite{cho:orl}, where the
underlying security assumption is the hardness of computing discrete
logarithms or factoring integers. 

Additionally, many papers have
introduced OT in the context of \textit{quantum}
cryptography~\cite{ben:bra:cre,cre:kil,dam:sal,may:sal,may,unr,yao},
where the legitimate users manipulate quantum states. The \emph{post-quantum} OT
research has gradually increased over time. 
Thus, some examples could be cited, such as the work done by~\cite{kazmi} and~\cite{vanessa:2019}. 

In general, in an $\binom{2}{1}$-OT protocol, a
sender sends two messages, say $m_{a}$ and $m_{b}$, and the receiver chooses only one of
them (for example, the receiver chooses $m_{a}$). At the end of the protocol, the sender does not know which of
the messages was chosen, and also the receiver learns nothing about
the other message (in this case, $m_b$). 

\textbf{Our contribuition}. According to~\cite{yl}, OT is one of the most important structures in cryptography for the construction of secure protocols. In terms of application, these types of protocols can be used in electronic auctions processes or even in contract signatures~\cite{gold} or electronic money transaction schemes~\cite{barak:2007}. In this work, our main objective was the implementation of $\binom{2}{1}$-SIOT protocol using the SIDH primitives from~\cite{feo:jao} for the purpose of providing privacy between sender and receiver, at the same time, providing a resistance against the imminent advent of the \emph{quantum} computation.

This paper is organized as follows: 
Section~\ref{sec:siot} describes our \emph{Supersingular Isogeny Oblivious Transfer} protocol. In section~\ref{sec:analyze}, We discuss about security aspects from~$\binom{2}{1}$- SIOT. In section~\ref{conclusion:sec}, we present a conclusion of the security analysis of the proposed protocol. An implementation of the proposed protocol is presented in section~\ref{implementation:siot}. In section~\ref{perform:ot}, we present  a performance estimate between some OT protocols. Finally, we conclude this work in section ~\ref{sec:conclusion}. In addition, we use appendix~\ref{apend:iso} to introduce some crucial background about isogenies of elliptic curves. A simplified presentation of the \emph{Supersingular Isogeny Diffie-Hellman} (SIDH) of~\cite{feo:jao} is shown in appendix~\ref{apendice:sidh}. Furthermore, in appendix~\ref{ot:protocol}, we present a brief concept of OT protocol and a simplified form of the protocol present in~\cite{cho:orl}. In appendix~\ref{linear:points}, we show some definitions about the process that determines linearly independent points used by our proposed protocol. In appendix~\ref{siot:sym}, there is the possibility of applying a symmetric pairing in the security analysis of the SIOT protocol. At last, in appendix~\ref{uv} we will verify that the proposed protocol is able to share certain points that allow to execute the OT functionality. 

\section{The $\binom{2}{1}$ - SIOT protocol}\label{sec:siot} 

In this section, we will see a new scheme called \emph{Supersingular Isogeny Oblivious Transfer} (SIOT) protocol. It is fundamentally inspired on schemes described in  \cite{cho:orl} and~\cite{feo:jao}.~For readers unfamiliar with issues of isogenies between elliptic curves and OT protocol, we suggest an initial reading in appendices~\ref{apend:iso},~~\ref{apendice:sidh} and~~\ref{ot:protocol}.

\subsection{Notations}\label{not2}

We use the cryptographic primitives of the \emph{Supersingular Isogeny Diffie-Hellman} key exchange protocol (SIDH) from \cite{feo:jao}. In this way, the following notations below will be used.

\begin{description}

\item[i.] $\mathcal{M}$, $\mathcal{K}$, $\mathcal{C}$ $\rightarrow$ Set of all plaintexts, keys and ciphertexts with binary strings of fixed length, respectively; 
\item[ii.] $p \rightarrow$ A prime such that  $p = 3 \bmod 4$;
\item[iii.] $\mathbb{F}_{p^{2}} \rightarrow$ A quadratic extension of $\mathbb{F}_{p}$, where $\mathbb{F}_{p^{2}}$ = $\mathbb{F}_{p}[i] / \langle i^{2} + 1\rangle$;
\item[iv.]$E_{0}[\mathbb{F}_{p^{2}}] \rightarrow$ A supersingular elliptic curve  over $\mathbb{F}_{p^{2}}$;
\item[v.]$ \mathbb{Z}/\ell \mathbb{Z} \rightarrow$ A field of integers modulo $\ell$, where $\ell$ is prime and $\ell$ $\nmid$ $p$;
\item[vi.]$P_{A}, Q_{A}  \rightarrow$ Linearly independent points over the supersingular elliptic curve  $E_{0}[\ell_A^{e_A}]$;
\item[vii.] $P_{B}, Q_{B}  \rightarrow$ Linearly independent points over the supersingular elliptic curve  $E_{0}[\ell_B^{e_B}]$;
\item[viii.] $\phi_{A}, \phi_{B} \rightarrow$ Isogenies  between $E_{0}$ and $E_{A}$,  $E_{0}$  and $E_{B}$, respectively;
\item[ix.] $\phi^{'}_{A}, \phi^{'}_{B} \rightarrow$ Isogenies  between $E_{B}$ and $E_{BA}$,  $E_{A}$  and $E_{AB}$, respectively;
\item[x.] $G_{A}, H_{A} \rightarrow$ Images of $P_{B}$ and $Q_{B}$ under Sender's private isogeny $\phi_{A}$; 
\item[xi.] $G_{B}, H_{B} \rightarrow$ Images of $P_{A}$ and $Q_{A}$ under Receiver's private isogeny $\phi_{B}$;
\item[xii.] $\jmath(E_{AB}), \jmath(E_{BA})  \rightarrow$  $\jmath-invariants$ of  supersingular elliptic curve $E_{AB}$ and $E_{BA}, respectively$;
\item[xiii] $e_A$, $e_B$ $\rightarrow$ Positive integers;
\item[xiv.] $r_{A}, r_{B}  \rightarrow$   Integers  from  $\mathbb{Z}/\ell^{e_A}_{A} \mathbb{Z}$  and  $\mathbb{Z}/\ell^{e_B}_{B} \mathbb{Z}$ , respectively; 
\item[xv.] $f$ $\rightarrow$ Small cofactor to ensure a prime $p = \ell^{e_{A}}_{A} \ell^{e_{B}}_{B} f \pm 1$;
\item[xvi.] $\mathcal{H} \rightarrow$ Hash function such that $\mathcal{H} = \{H_{k} : k \in K\}$. It is a hash function family indexed by a finite set $K$,  where each $H_{k}$ is a function from $\mathbb{F}_{p^{2}}$.
\item[xvii.]$\mathcal{E}$, $\mathcal{D}$ $\rightarrow$ Encryption and Decryption algorithms, respectively; 
\item[xviii.]$\pk_A$, $\pk_B$ $\rightarrow$ Sender's public key and Receiver's public key, respectively;
\item[xiv.] $\sk_A$, $\sk_B$ $\rightarrow$ Sender's private key and Receiver's private key, respectively.
\item[xx.] $\mathcal{\omicron}$ $\rightarrow$ Special point located at infinity. It acts as a neutral element in the operation of adding points of an elliptic curve over a finite field;
\item[xxi.] $ker(\phi)$ $\rightarrow$ Kernel of an isogeny. In particular, it is a finite subgroup of an elliptic curve over closed field $\mathbb{F}_{p^{2}}$. It can also be denoted by $\langle P, Q \rangle$, where $P , Q \in E[\mathbb{\bar{F}}_{p^{2}}]$.

\end{description}

\subsection{Public parameters} \label{parameters}

Let $E_0$ be a supersingular elliptic curve defined over $\mathbb{F}_{p^2}$. For convenience, assume a prime $p$ of form\footnote{See \cite{jao:reza} that reports a deep research about  the choice of SIDH-Friendly Primes.} $p = \ell_A^{e_A} \ell_B^{e_B} - 1$ with $\ell_A = 2$ and $e_A \geqslant 4$ (and $f = 1$), or $p = 4\ell_A^{e_A}\ell_B^{e_B} - 1$, where both $\ell_A$ and $\ell_B$ are odd primes (and $f = 4$). Hence, either of these choices yield $p = 3 \pmod 4$, enabling the representation\footnote{See \cite{hoff:pip:silv} for more details about this type of representation.} $\mathbb{F}_{p^2} = \mathbb{F}_p[i] / \langle i^2 + 1 \rangle$ and ensuring that the curve $E_0 [ \mathbb{F}_{p^2}]: y^2 = x^3 + x$ is supersingular, with group order $(\ell_A^{e_A}\ell_B^{e_B}f)^2$. Additionally, let $P_A, Q_A, P_B, Q_B$ points. Then, $E_0[\ell_A^{e_A}]$ and $E_0[\ell_B^{e_B}]$ are generated by kernel  $\langle P_A, Q_A \rangle$ and $\langle P_B, Q_B \rangle$, respectively. The appendix~\ref{linear:points} presents some definitions used to compute such linearly dependent points.

\subsection{Premises}\label{def:premissas} 

\begin{enumerate}

\item Let ($\mathcal{E}, \mathcal{D}$) be a symmetric encryption scheme according to definitions 1 and 2 of~\cite{cho:orl}. The shared symmetric key is the value of the invariant $\jmath$ between two supersingular and isomorphic\footnote{See~\cite{silv}, Proposition III.1.4(b).} elliptic curves. We will see in figure~\ref{fig:SIOT}, that the invariant $\jmath$ will be submitted to hash function $\mathcal{H} = \{ H_{k} : k \in \mathcal{K} \}$ indexed by a finite set $\mathcal{K}$, where each $H_{k}$ is a function of $\mathbb{F}_{p^{2}}$;

\item Alice wants to encrypt two messages $m_{0}$, $m_{1}$ $\in \mathcal{M}$ and send them to Bob. In turn, Bob will decrypt only one of these two messages and Alice will not be aware of his choice; 

\item Alice and Bob use a \emph{coin-flipping} protocol from~\cite{wagner:2016} to share a single uniform random string of $w$ bits. This ensures that neither Alice nor Bob can guess in advance or control the value of $w$. Thus, they must use, for instance, a hash function in this bit string to get the linearly independent points $U$ and $V$. Otherwise, they must generate a new string $w$. 

\end{enumerate}

\subsection{Protocol}
 
 Figure~\ref{fig:SIOT} shows an abstraction of the proposed protocol operation in its simplest form, i.e, $\binom{2}{1}$ - SIOT. 
 
 \begin{figure}[h]
 \centering
  \resizebox{0.99\textwidth}{!}{%
 
 \fbox{\pseudocode{%
  \> \textbf{Supersingular Isogeny Oblivious Transfer (SIOT)}\> \\ [0.1\baselineskip][\hline] 
 \textbf{Sender} \> \> \textbf{Receiver} \\
  \text{Input: } m_0, m_1 \in \mathcal{M} \>\> \text{Input: } \sigma \in \set{0, 1} \\
  \text{Output: None} \>\> \text{Output: } m_\sigma\\
  r_A \leftarrow \ZZ/\ell_A^{e_A}\ZZ  \>\> r_B \leftarrow \ZZ/\ell_B^{e_B}\ZZ\\
   \text{$\phi_A: E_{0} \rightarrow E_{A}$}\>  \> \text{$\phi_B: E_{0} \rightarrow E_{B}$} \\ 
  \text{$E_{A} \leftarrow E_{0} / \langle P_A + r_{A}Q_{A}\rangle$}\>  \> \text{$E_{B} \leftarrow E_{0} / \langle P_B + r_{B}Q_{B}\rangle$} \\
    G_A \leftarrow \phi_A(P_B); H_A \leftarrow \phi_A(Q_B) \>\> G_B \leftarrow \phi_B(P_A); H_B \leftarrow \phi_B(Q_A)  \\
  \pk_A \leftarrow (E_A, G_A, H_A) \>\> \pk_B \leftarrow (E_B, G_B, H_B)  \\
  \> \sendmessageright{centercol=3,top=$\pk_A$} \> \text{If } G_A, H_A \notin E_A[\ell_B^{e_B}], \\
  \> \> \quad \text{then abort ($\perp$).} \\
   \>\>   U, V \sample  E_B[\ell_A^{e_A}] \; \text{then},\\
  \>\> \hat{G}_B \leftarrow (G_B - \sigma U)\\
  \>\> \hat{H}_B \leftarrow (H_B - \sigma V)\\
   \> \> \hat{\pk}_B \leftarrow (E_B, \hat{G}_B, \hat{H}_B) \\
  \text{If } \hat{G}_B, \hat{H}_B \notin E_B[\ell_A^{e_A}], \> \sendmessageleft{centercol=3,top=$\hat{\pk}_B$} \>  \\
  \quad \text{then abort ($\perp$).} \>\> \\ 
  \forall i \in \set{0, 1} \text{and}\;   U, V \sample  E_B[\ell_A^{e_A}] \; \text{then}, \>\> \text{$\phi^{'}_{B}: E_{A} \rightarrow E_{AB}$} \\ 
   \text{$\phi^{'}_{A_{i}}: E_{B} \rightarrow E_{BA_{i}}$}\>\>\\
  E_{BA_i} \leftarrow \langle (\hat{G}_B + iU) + r_A(\hat{H}_B + iV) \rangle \>\> \text{$E_{AB} \leftarrow E_{A}/ \langle G_{A} + r_{B}H_{A}\rangle$} \\ 
  k_{0} = \mathcal{H} ( \jmath(E_{BA_{0}})) \>\>  k_{\sigma} = \mathcal{H}(\jmath(E_{AB})) \\ 
  k_{1} = \mathcal{H} (\jmath(E_{BA_{1}})) \>\>  \\
  \>\sendmessageright{centercol=3,top=$c_0 \gets \mathcal{E}_{k_0}(m_0)$, bottom=$c_1 \gets \mathcal{E}_{k_1}(m_1)$} \>  \\
  \< \<   \\
  \>\> m_\sigma \gets \mathcal{D}_{k_{\sigma}}(c_\sigma);\\
   \< \<  \\
   \< \<   \\ 
  \<\text{\textbf{}}\<
 \>\> } 

 } 

  }
  
 \caption{$\binom{2}{1}$ - SIOT protocol.}
 \label{fig:SIOT}
  
 \end{figure} 

\subsubsection{Generation of key pairs}

\paragraph*{Setup - Sender}

\begin{enumerate}
\item Alice secretly chooses a value  $r_{A} \leftarrow Z / \ell_A^{e_A}Z $;

\item She computes:

\begin{enumerate}

\item $ker(\phi_{A}) =  \langle P_{A} + r_{A}Q_{A} \rangle$;

\item $\phi_{A} : E_{0} \rightarrow E_{A}$;

\item $\phi_{A}(P_{B}) = G_{A}$; $\phi_{A}(Q_{B}) = H_{A}$.

\end{enumerate}

\item Alice creates a pair of keys $sk_{A} = (\phi_{A} , r_{A})$ and $pk_{A} = (E_{A}, G_{A}, H_{A})$, i.e, her private and public keys, respectively;

\item Alice sends to Bob $pk_{A} = (E_{A}, G_{A}, H_{A} )$. He checks if  $G_{A}, H_{A} \in E_{A}[ \ell_B^{e_{B}}]$, i.e, $\ell^{e_{B}}_{B} G_{A} = \ell^{e_{B}}_{B} H_{A} = \mathcal{\omicron}_{A} \in E_{A}$, since $\{P_{B}, Q_{B}\} \subset E_{0}[\ell^{e_{B}}_{B}]$. If this check is valid then, Bob will accept the public key of Alice. Otherwise, that public key will be rejected ($\perp$).

\end{enumerate}

It is denoted $E_A = E_0 / \langle P_A + r_A Q_A \rangle$, where $\vert ker(\phi_A) \vert = \vert \langle P_A + r_A Q_A \rangle \vert =   \ell_A^{e_A}$, $i.e$, a separable\footnote{See~\cite{law}, Proposition 12.8} isogeny of degree $\ell^{e_{A}}_{A}.$ 

\paragraph*{Setup - Receiver}

\begin{enumerate}
\item Bob secretly chooses a value  $r_{B} \leftarrow Z / \ell_B^{e_B}Z $; 

\item He computes:

\begin{enumerate}

\item $ker(\phi_{B}) =  \langle P_{B} + r_{B}Q_{B} \rangle$;

\item $\phi_{B} : E_{0} \rightarrow E_{B}$;

\item $\phi_{B}(P_{A}) = G_{B}$; $\phi_{B}(Q_{A}) = H_{B}$; 

\item For a unique $\sigma \in \bin$, $\hat{G}_{B} = (G_{B} - \sigma U)$ and $\hat{H}_{B} = (H_{B} - \sigma V)$.

\end{enumerate}

\item Bob creates a key pair $sk_B = (\phi_{B}, r_{B})$ and $\hat{pk}_B = (E_{B}, \hat{G}_{B}, \hat{H}_{B})$, i.e, his private and public key, respectively;

\item He sends $\hat{pk}_B$ to Alice. Then, she performs two checks:

\begin{itemize}
\item If $\hat{G}_{B}, \hat{H}_{B}$ $ \in E_{B}[\ell_A^{e_A}]$, i.e, $\ell_A^{e_A} G_B = \ell_A^{e_A} H_B = \mathcal{\omicron}_B \in E_B$ since $\{P_A, Q_A\} \subset  E_0[\ell_A^{e_A}]$;
\item If the points $U$, $V$ $ \in E_{B}[\ell_A^{e_A}]$. This ensure that the pair of points $(G_{B} + U, H_{B} + V)$ and $(G_{B} - U, H_{B} - V)$ are generated by $E_{B}[\ell_A^{e_A}]$. Otherwise, $\hat{pk}_B$ will be rejected ($\perp$) and the protocol is restarted to execute another bit string $w$.  
\end{itemize}

\end{enumerate}

\subsubsection{Generation of secret keys}\label{gen:key}

\paragraph*{Setup - Sender}

\begin{enumerate}

\item Alice computes:

\begin{enumerate}
\item $\forall i \in \bin$, $ker(\phi'_{A_{i}})$ = $\langle (\hat{G}_{B} + iU) + r_{A}(\hat{H}_{B} + iV) \rangle$;

\item $\phi'_{A_{i}} : E_{B} \rightarrow E_{BA_{i}}$;

\item $\jmath_{i} \leftarrow \jmath(E_{BA_{i}});$

\item $k_{i} = \mathcal{H}( \jmath_{i}).$

\end{enumerate}

\end{enumerate}

\paragraph*{Setup - Receiver}

\begin{enumerate}

\item Bob computes: 

\begin{enumerate}
\item $ker(\phi'_{B}) = \langle G_{A} + r_{B}H_{A} \rangle$;

\item $\phi'_{B} : E_{B} \rightarrow E_{AB}$;

\item $\jmath_{\sigma} \leftarrow j(E_{AB})$;

\item $k_{\sigma} = \mathcal{H}(\jmath_{\sigma}).$ 
\end{enumerate}

\end{enumerate}

\subsubsection{Encryption and Decryption}

\begin{enumerate}

\item $\forall i \in \bin$, Alice encrypts $m_{i}$. Then, $c_{i} \leftarrow \mathcal{E}( k_{i}, m_{i} )$. After that, she sends $(c_{0}, c_{1})$ to Bob;

\item He decrypts and gets $m_{\sigma} \leftarrow \mathcal{D} (k_{\sigma}, c_{\sigma})$;

\end{enumerate}

\section{Security analysis of the $\binom{2}{1}$-SIOT protocol}\label{sec:analyze} 

In the next sections, we will present some definitions  for security analysis of the proposed protocol. 

\subsection{Preliminaries}\label{pre:pre}

\begin{definition}\label{pre:negli}

A function $\epsilon(\cdot)$ is negligible in $\underline{n}$, or just negligible, if for every positive polynomial $p(\cdot)$ and all sufficiently large $\underline{n}$ it holds that $\epsilon(\cdot) < 1 / p(\cdot)$. 

\end{definition}

\begin{definition}

A probability ensemble $\mathcal{X} = \{ \mathcal{X}(n, a)\}$ is an infinite sequence of random variables indexed by  $n \in \mathbb{N}$ and $a \in \bin^{*}$. The value $\underline{n}$ will represent a security parameter and $\underline{a}$ will represent the parties' inputs. 

\end{definition}

\begin{definition}\label{pre:indist}

Two distribution ensembles $\mathcal{X} = \{ \mathcal{X}(n, a)\}$ and $\mathcal{Y} = \{ \mathcal{Y}(n, a)\}$ are said to be computationally indistinguishable, denoted by $ \mathcal{X}$ $\overset{\mathbf{c}}{\equiv}$ $\mathcal{Y}$, if for every non-uniform polynomial-time algorithm $\mathcal{D}$ there exists a negligible function $\epsilon(\cdot)$ such that for every $n \in \mathbb{N}$ and $a \in \bin^{*}$. Then, $\vert \Pr [ \mathcal{D} (\mathcal{X}(n, a)) = 1] - \Pr [ \mathcal{D} (\mathcal{Y}(n, a)) =1] \vert \leqslant \epsilon(n)$.

\end{definition} 

\subsection{Computational problems of isogenies between supersingular elliptic curves\\}\label{problem}

In this section, we will see some cases of computational problems from supersingular elliptic curves that were adapted from  \cite{feo:jao}. Therefore, let a supersingular curve $E_{0}$ over $\mathbb{F}_{p^2}$  together with independent bases $\{ P_{A}, Q_{A} \}$ and $\{ P_{B}, Q_{B} \}$ of $E_0[\ell_A^{e_A}]$ and $E_0[\ell_B^{e_B}]$, respectively.  Furthermore, recall that $p$ is a prime of the form defined on section~\ref{parameters}.

\begin{problem}[Decisional Supersingular Isogeny (DSSI) problem]\label{def:dssi}
Let $E_A [\mathbb{F}_{p^2}]$ be another supersingular curve. Decide whether $E_A$ is $\ell_A^{e_A}$-isogenous to $E_0$. 

\end{problem}

\begin{problem}[Computational Supersingular Isogeny (CSSI) problem]\label{def:cssi}
Let $\phi_A: E_0 \rightarrow E_A$ be an isogeny whose kernel is $R_{A}$ = $\langle [m_A]P_A + [r_A]Q_A\rangle$ for some $m_A, r_A \in \mathbb{Z}/ \ell_A^{e_A}\mathbb{Z}$. Given the public key $(E_A, G_A, H_A)$. Determine $R_{A}$. 
\end{problem}

\begin{problem}[Supersingular Computational Diffie-Helmann (SSCDH) problem]\label{prob:sscdh}
Let $\phi_A: E_0 \rightarrow E_A$ be an isogeny whose kernel is $R_A$ = $\langle [m_A]P_A + [r_A]Q_A\rangle$ for some $m_A, r_A \in \mathbb{Z}/ \ell_A^{e_A}\mathbb{Z}$ and let $\phi_B: E_0 \rightarrow E_B$ be an isogeny whose kernel is $R_B$ = $\langle [m_B]P_B + [r_B]Q_B\rangle$ for some $m_B, r_B \in \mathbb{Z}/ \ell_B^{e_B}\mathbb{Z}$. Given the public keys ($E_A, G_A, H_A $) and ($E_B, G_B, H_B$). Determine the $\jmath - invariant$ of $E_0 / \langle [m_A]P_A + [r_A]Q_A , [m_B]P_B + [r_B]Q_B \rangle$

\end{problem}

\begin{problem}[Supersingular Decision Diffie-Hellman (SSDDH) problem]\label{prob:ssddh} 
Given a tuple sampled with probability $1/2$ from one of the following two distributions:

\begin{enumerate}
\item $(E_A, E_B, G_A, H_A , G_B, H_B, E_{AB})$ where $E_A, E_B, G_A , H_A, G_B \, and \, H_B$ are as in the SSCDH problem (\textit{Problem~\ref{prob:sscdh}}) then,

\begin{center}

$E_{AB} \simeq E_0/\langle [m_A]P_A + [r_A]Q_A, [m_B]P_B + [r_B]Q_B \rangle$,

\end{center}

\item $(E_A, E_B, G_A, H_A, G_B, H_B, E_{C})$ where $E_A, E_B, G_A , H_A, G_B \, and \, H_B$ are as in the SSCDH problem (\textit{Problem~\ref{prob:sscdh}}) then,

\begin{center}

$E_{C} \simeq E_0/\langle [\hat{m}_A]P_A + [\hat{r}_A]Q_A, [\hat{m}_B]P_B + [\hat{r}_B]Q_B \rangle$

\end{center}
\end{enumerate}

Let $m_A, r_A, \hat{m}_A, \hat{r}_A  \in \mathbb{Z}/\ell_A^{e_A}\mathbb{Z}$ and $m_B, r_B, \hat{m}_B, \hat{r}_B  \in \mathbb{Z}/\ell_B^{e_B}\mathbb{Z}$. Determine from which distribution the tuple is sample.

\begin{remark}
Each sample has a probability 1/2. Thus, for definition~\ref{pre:indist}, we have that:

\begin{center}

$\{E_A, E_B, G_A, H_A , G_B, H_B , E_{AB}\}  \overset{\mathbf{c}}{\equiv} \{E_A, E_B, G_A, H_A , G_B, H_B , E_{C}\}$ \\ 

\end{center}

\end{remark}

\end{problem}

\subsection{Notations for security analysis}

In the security analysis of the proposed protocol, the followings notations will be used: 

\begin{description}
\item[i.]Application of the $V\acute{e}lu's\footnote{See~\cite{gal}, Corollary 25.1.7.} \: formula \to V\acute{e}lu's \: formula\{ker(\phi), E\}$;

 \item[ii. ] According to~\cite{kalai:2005} in an OT protocol, the replacement of either $m_0$ or $m_1$ by another message $m$ must go unnoticed by the receiver. Let $\tau$ $\in \mathcal{M}$ and $\sigma \in \bin$. Then, Alice's view in executing an OT protocol is denoted by $\{\Omega_{Alice}(Alice(1^{n}, \tau), Bob(1^{n}, \sigma))\}$, where a security parameter is defined by $1^{n}$. Similarly, we denote Bob's view for  $\{\Omega_{Bob}(Alice(1^{n}, \tau), Bob(1^{n}, \sigma))\}$; 
\item[iii. ] When Alice or Bob acts like a dishonest user, we denote them by Alice* and Bob*, respectively.  
\end{description}

\subsection{Some requirements for security analysis}\label{sec:security} 

\emph{A priori}, any secure protocols should resist to any adversarial attack. Thus, to prove that an OT  protocol is secure,~\cite{yl} state that the most important requirements in any security protocol are \emph{correctness} and \emph{privacy}. 

\subsubsection{Correctness}$ $\label{sec:correctness}

Suppose that both Alice and Bob are honest parties taking the $\binom{2}{1}$-SIOT  protocol. Let $\sigma, i \in \bin$ such that $\sigma = i$. Thus, the \textit{correctness} follows the identities below. 
\begin{align*}
 \jmath(E_{BA_{i}}) \simeq& \jmath(E_{B} / \langle (G_{B} - \sigma \cdot U + i \cdot U) + r_{A} \cdot (H_{B} - \sigma \cdot V + i \cdot V) \rangle) \\
                               \simeq& \jmath(\phi_{B}(E_{0}) / \langle(\phi_{B}(P_{A}) - \sigma \cdot U + i \cdot U) + r_{A}(\phi_{B}(Q_{A}) - \sigma \cdot V + i \cdot V) \rangle)\\
                               \simeq& \jmath(\phi_{B}(E_{0}) / \langle \phi_{B}(P_{A}) + r_{A} \cdot \phi_{B}(Q_{A}) \rangle)\\
                               \simeq& \jmath(\phi_{B}(E_{0}) / \langle \phi_{B}(P_{A} + r_{A} \cdot Q_{A}) \rangle)\\
                               \simeq& \jmath(\phi'_{A_{i}}(\phi_{B}(E_{0}))) \simeq \jmath(\phi_{A}(E_{0}) / \langle \phi_{A}(P_{B} + r_{B} \cdot Q_{B}) \rangle)\\
                               \simeq& \jmath(\phi'_{A_{i}}(\phi_{B}(E_{0}))) \simeq \jmath(E_{B}/  \langle \phi_{A}(P_{B}) + r_{B} \cdot \phi_{B}( Q_{B}) \rangle )\\
                               \simeq& \jmath(\phi'_{A_{i}}(\phi_{B}(E_{0}))) \simeq \jmath(E_{B}/  \langle G_{A} + r_{B} \cdot H_{A} \rangle )\simeq \jmath(\phi'_{B}(\phi_{A}(E_{0}))) \simeq \jmath(E_{AB}).
\end{align*}

\subsubsection{Privacy}$ $

In the $\binom{2}{1}$-SIOT protocol, Bob's choice should not be known to Alice. Moreover, at the end of the protocol execution, Bob will not be able to gain any knowledge about the message that he did not decrypt. It should be noted that this privacy stems from the difficulty of solving the computational problems seen in section~\ref{problem}. Finally, to complement the security proof of the proposed protocol, Theorem~\ref{theorem} was elaborated as follows:

\begin{theorem}\label{theorem}

Assume that CSSI, SSCDH, SSDDH problems are hard in a group $E(\mathbb{F}_{p^{2}})$. Then, $\binom{2}{1}$-SIOT protocol ensures privacy between two parties. 


\begin{proof}

The proof is to adapt definition 2.6.1 of~\cite{yl} to $\binom{2}{1}$-SIOT protocol for compatibility with the computational problems mentioned in section~\ref{problem}. Let two messages, $m_0$ and $m_1$, between two parts (Alice and Bob). An OT protocol is private if the following requirements are valid: 

\begin{description}
\item[i.] \textbf{Non-triviality}: If Alice and Bob follow the protocol correctly then, after an execution in which Alice has for input any $m_{0}$, $m_{1}$ $\in$ $\mathcal{M}$ and Bob has input bit $\sigma$ $\in$ $\bin$, the output of Bob is $m_{\sigma}$. In other words, Bob receives $pk_{A}$ and the pair $(c_{0}, c_{1})$ from Alice. Recalling that $pk_{A} \leftarrow (E_{A}, G_{A}, H_{A})$  and $c_{\sigma} \leftarrow \mathcal{E} (k_{\sigma}, m_{\sigma})$ are well defined. Thus, non-triviality follows from the fact that 

\begin{center}

$V\acute{e}lu's \, formula\{ \langle G_A + r_B H_A \rangle  , E_{A} \}$ $\Rightarrow$ $E_{AB}$ $\therefore$ \\

$\mathcal{H}(\jmath(E_{AB}))$ $\Rightarrow$ $k_{\sigma}$. 

\end{center}

Therefore, Bob recovers $k_{\sigma}$ implying $m_{\sigma} \leftarrow \mathcal{D}(k_{\sigma}, c_{\sigma})$ such that $\sigma$ is a unique binary value secretly chosen by him. Furthermore, upon receiving $pk_{A}$, Bob will not be able to compute Alice's private key, i.e, $\sk_A$ = $(\phi_{A}, r_{A})$. If that could be possible, there would be a violation of the CSSI problem difficult hypothesis.\\

\item[ii.] \textbf{Privacy in the case of a dishonest Bob:} Let $\hat{pk}_{B} \leftarrow (E_B, \hat{G}_B, \hat{H}_B)$ denotes Bob's public key sent to Alice. Recall that $\hat{G}_{B} \leftarrow G_{B}$, $\hat{H}_{B} \leftarrow H_{B}$, if $\sigma = 0$ and $\hat{G}_{B} \leftarrow( G_{B} - U)$, $\hat{H}_{B} \leftarrow (H_{B} - V)$, if $\sigma = 1$. In addition, there is Alice's public key $pk_{A} := (E_{A}, G_{A}, H_{A})$ sent to Bob and a unique value of  $\jmath - invariant$ $\jmath_{\sigma} = \jmath(V\acute{e}lu's \, formula\{ \langle G_A + r_B H_A \rangle  , E_{A} \})$ computed by Bob upon receiving Alice's well-defined public key $pk_{A}$. After that, $\forall i \in \bin$, Alice will compute $\jmath_{i} = \jmath(V\acute{e}lu's \, formula\{ \langle ( \hat{G}_B + iU) + r_A (\hat{H}_B + iV) \rangle  , E_{B} \})$, $i.e$, $\jmath_{0}$ and $\jmath_{1}$. Moreover, Alice will share a unique secret key with Bob. Thus, Alice's privacy is based on the following fact: Bob cannot compute both values of the invariants $\jmath_{0}$ and $\jmath_{1}$ ($\jmath_{0} \neq \jmath_{1}$), if the hypothesis of the SSCDH problem is difficult. In other words, Bob will be able to compute a unique invariant $\jmath_{\sigma}$.  


%
%
%
%
%

Let $\sigma \in \bin$ an auxiliary input and another input with tuple $m_{0}$, $m_{1}$, $m$ $\in \mathcal{M}$. Thus, another way to view the Alice's privacy is that Bob's first message, denoted by Bob* $(1^{n}, \sigma)$, determines whether it should receive $m_{0}$ or $m_{1}$. For example, if it determines that it should receive $m_{0}$, then its view when Alice's input is $(m_{0}, m_{1})$ is indistinguishable from its view when Alice's input is $(m_{0}, m)$. Evidently, this implies that Bob cannot learn anything about $m_{1}$ when it receives $m_{0}$ and \emph{vice versa}. Hence, \\ 

\begin{equation*}
\resizebox{0.91\hsize}{!}{%
$\{\Omega_{Bob^{*}} (Alice(1^{n}, (m_{0}, m_{1})); Bob^{*}(1^{n}, \sigma))\}_{n \in \mathbb{N}} \overset{\mathbf{c}}{\equiv} \{\Omega_{Bob^{*}} (Alice(1^{n}, (m_{0}, m)); Bob^{*}(1^{n}, \sigma))\}_{n \in \mathbb{N}}$
}
\end{equation*}
\begin{center}
or
\end{center}
\begin{equation*}
\resizebox{0.91\hsize}{!}{%
$\{\Omega_{Bob^{*}} (Alice(1^{n}, (m_{0}, m_{1})); Bob^{*}(1^{n}, \sigma))\}_{n \in \mathbb{N}}\overset{\mathbf{c}}{\equiv}\{\Omega_{Bob^{*}} (Alice(1^{n}, (m, m_{1}));  Bob^{*}(1^{n}, \sigma))\}_{n \in \mathbb{N}}$ 
}
\end{equation*}

\item[iii.] \textbf{Privacy in the case of a dishonest Alice:} Note that this requirement shows that Alice cannot distinguish Bob's possible secret choices, i.e, when bit $\sigma$ is set to 0 or 1. In other words, she simply visualizes a $pk_B$ public key. Then, the receive's privacy will be checked by following Lemma: 

\begin{lemma}\label{sec:lemma1}

Alice by inputting $\hat{pk}_{B}$ cannot guess $\sigma$ with probability greater than $1 / 2 + \epsilon(n)$, for some negligible function $\epsilon(n)$ and $\forall$ n $\in \mathbb{N}$

\begin{proof}

We can assume that by receiving $\hat{pk}_{B}$ and not knowing the value of Bob's bit $\sigma$, Alice cannot distinguish the pairs of tuples $\{(E_{B}, G_{B}, H_{B})\}$ and $\{(E_{B}, (G_{B} - U), (H_{B} - V)\}$, $i.e$, for some $\hat{G}_{B}, \hat{H}_{B} \in E_{B}[\ell^{e_{A}}_{A}]$  such that  $\Pr[(E_{B}, {G}_{B}, {H}_{B}) = (E_{B}, \hat{G}_{B}, \hat{H}_{B})]$ = $\Pr[(E_{B}, {G}_{B} - U, {H}_{B} - V) = (E_{B}, \hat{G}_{B}, \hat{H}_{B}))]$ which is independent of $\sigma$ . Thus, the difficulty of this indistinguishability between these tuples is based on SSDDH problem. We have that:

\begin{center}

$\{ (E_{B}, G_{B}, H_{B})\}$ $\overset{\mathbf{c}}{\equiv}$ $\{E_{B}, ({G}_{B} - U), ({H}_{B} - V) \}$

\end{center}

According to~\cite{phi}, a distinguisher is a probabilistic algorithm that describes the advantages of an adversary's advantage. Then, suppose that, by contradiction, there is a probabilistic distinguisher $\Theta$ of polynomial time and a non-negligible function $\epsilon$ such that $\forall n \in \mathbb{N}$,

\begin{center}

$\vert P_{r} [\Theta(E_{B}, G_{B}, H_{B}) = 1] -  P_{r} [\Theta(E_{B}, ({G}_{B} - U), ({H}_{B} - V) \}) = 1] \vert$ $\geqslant$ $\epsilon(n)$,

\end{center}  

Then, by subtracting and adding the following term,

\begin{center}

$P_{r} [\Theta(E_{B}, ({G}_{B} - U - R), ({H}_{B} - V - S)) = 1]\footnote{$ R, S  \in E_{B}[\ell^{e_{A}}_{A}]$.}$

\end{center}

We have that

$\vert P_{r} [\Theta(E_{B}, G_{B}, H_{B}) = 1] -  P_{r} [\Theta((E_{B}, ({G}_{B} - U), ({H}_{B} - V)) = 1] \vert \leqslant \vert P_{r} [\Theta(E_{B},\\ G_{B}, H_{B}) = 1] - P_{r} [\Theta(E_{B}, ({G}_{B} - U - R), ({H}_{B} - V - S)) = 1]\vert + \vert P_{r} [ \Theta(E_{B}, ({G}_{B} - U - R), ({H}_{B} - V - S)) = 1] -  P_{r} [ \Theta(E_{B}, ({G}_{B} - U), ({H}_{B} - V)) = 1] \vert $ \\

By contradiction, We suppose that 

\begin{equation}\label{eq:01}
\resizebox{0.85\hsize}{!}{%
$\vert P_{r} [\Theta(E_{B}, G_{B}, H_{B}) = 1] -  P_{r} [\Theta(E_{B}, (G_{B} - U - R), (H_{B} - V - S)) = 1]  \vert \geqslant \dfrac{\epsilon(n)}{2}$  
}
\end{equation}
\begin{center}
or
\end{center}
\begin{equation}\label{eq:02}
\resizebox{0.85\hsize}{!}{%
$\vert P_{r} [\Theta(E_{B}, (G_{B} - U - R), (H_{B} - V - S)) = 1] -  P_{r} [\Theta(E_{B}, (G_{B} - U), (H_{B} - V)) = 1]  \vert \geqslant \dfrac{\epsilon(n)}{2}$  
}
\end{equation}

Suppose that (3.1) holds. Thus, we can construct a distinguisher $\tilde{\theta}$ for the SSDDH problem  that works as follow: Upon input $\hat{p}k_{B} \leftarrow \{ (E_{B}, ({G}_{B} - U), ({H}_{B} - V)) \}$, the distinguisher $\tilde{\Theta}$ randomly chooses the pair of points $R, S$. Hence, $\hat{p}k'_{B} \leftarrow \{ (E_{B}, ({G}_{B} - U - R), ({H}_{B} - V - S))\}$. On the other hand,  if $\hat{p}k_{B} \leftarrow \{(E_{B}, G_{B}, H_{B}) \}$ then, $\hat{p}k'_{B} \leftarrow \{E_{B},  (G_{B} - R), (H_{B} - S) \}$. Note that the pairs of points $U$ and $V$ are not used in the last tuple $\hat{p}k'_{B}$. However, these points as points $R$ and $S$ from to the same group $E_{B}[\ell^{e_{A}}_{A}]$ and could also be randomly chosen by $\tilde{\Theta}$, say points $U$ and $V$. Thus, we can have that $\hat{p}k'_{B} \leftarrow \{(E_{B},G_{B}, H_{B})\}$ and\\

$\vert P_{r} [\tilde{\Theta}(E_{B}, G_{B}, H_{B}) = 1] -  P_{r} [\tilde{\Theta}(E_{B}, ({G}_{B} - U), ({H}_{B} - V)) = 1] \vert =  \vert P_{r} [\Theta(E_{B},\\ G_{B}, H_{B}) = 1] -  P_{r} [\Theta(E_{B}, ({G}_{B} - U - R), ({H}_{B} - V- S)) = 1]  \vert \geqslant \dfrac{\epsilon(n)}{2}$,\\

in contradiction to the SSDDH problem. An analogous analysis follows in the case where (3.2) holds. The proof of Bob's privacy is concluded by noting that $\{(E_{B}, G_{B}, H_{B})\}$, $\{(E_{B}, ({G}_{B} - U), ({H}_{B} - V))\}$, regardless of the value of $\sigma$, are indistinguishable in Alice's view. In other words, let $\tau \in \bin^{*}$ be an auxiliary input. Thus, \\

\begin{center}

$\{\Omega_{Alice^{*}} (Alice^{*} (1^{n}, \tau), Bob(1^{n}, 0))\}$  $\overset{\mathbf{c}}{\equiv}$ $\{\Omega_{Alice^{*}} (Alice^{*} (1^{n}, \tau), Bob(1^{n}, 1))\}$.

\end{center}

According to Lemma~\ref{sec:lemma1}, the privacy of Bob follows from SSDDH problem over the group $E_{B}[\ell^{e_{A}}_{A}]$. \qed

\end{proof}
\end{lemma}

\end{description}

Therefore, the above requirements are related to the computacional problems of isogenies in supersingular elliptic curves and they ensure the privacy of the $\binom{2}{1}$-SIOT protocol  \qed
\end{proof}
\end{theorem}

\vspace{-3cm}

\subsection{Algebraic security analysis of the $\binom{2}{1}$-SIOT protocol}
\vspace{-2cm}
      Considering the case of a dishonest Alice, she will use a \emph{Weil} pairing-based distinguisher  for  trying to find out the secret value $\sigma$ from honest Bob. In the second situation, the roles will be inverted, $i.e$, Alice will be considered an honest sender and Bob a dishonest receiver. Thus,  an analysis is performed in such a way that some algebraic conditions must be obeyed so that Bob is not able to decipher both of Alice's messages.
      
\vspace{-3cm}

\subsubsection{Preventing a \emph{Weil} pairing-based distinguisher from a possible Alice's dishonesty}\label{sec:indist} $ $

Considering the situation where Alice* (the dishonest sender) receiving the information $(E_{B}, \hat{G}_{B}, \hat{H}_{B})$ from Bob, \emph{a priori}, does not know whether to receive  $(E_{B}, G_{B}, H_{B})$ or $(E_{B}, G_{B} - U, H_{B} - V)$. Alice might consider using the \emph{Weil} pairing to distinguish between these two values. 
In what follows, all pairings have order $\ell_A^{e_A}$. After all, the correct points $G_B = \phi_B(P_A)$ and $H_B = \phi_B(Q_A)$ are known to satisfy $e(G_B, H_B) = e(P_A, Q_A)^{\ell_A^{e_A}}$: if this relation does not hold for both of $(\hat{G}_B, \hat{H}_B)$ or $(\hat{G}_B + U, \hat{H}_B + V)$, it would reveal which key Bob has chosen.
More generally, because Alice can add any multiple of $(U, V)$ to $(\hat{G}_B, \hat{H}_B)$ and look for such a mismatch, one must have $e(\hat{G}_B + \lambda U, \hat{H}_B + \lambda V) = e(G_B, H_B) = e(P_A, Q_A)^{\ell_A^{e_A}}$ for any $\lambda \in\mathbb{Z}/ {\ell_A^{e_A}} \mathbb{Z}$. Recalling that $U, V, G_{B}, H_{B} \in E_{B}[\ell^{e_{A}}_{A}]$ then, they can be wrote as a linear combination, i.e, $U = \alpha G_B + \beta H_B$, $V = \gamma G_B + \delta H_B$ such that $\alpha, \beta, \gamma, \delta$ $\in \mathbb{Z} / \ell^{e_{A}}_{A} \mathbb{Z}$. Thus, this condition means that:

\pagebreak

\begin{align*}
e(G_B + \lambda U, H_B + \lambda V) &=     e(G_B + \lambda \alpha G_B + \lambda \beta H_B, H_B + \lambda \gamma G_B + \lambda \delta H_B)\\
                    &=          e((1 + \lambda \alpha)G_B + \lambda \beta H_B, \lambda \gamma G_B + (1 + \lambda \delta) H_B)\\
                    &=          e((1 + \lambda \alpha)G_B, \lambda \gamma G_B)\\
                    &\;\cdot\,\,e((1 + \lambda \alpha)G_B, (1 + \lambda \delta) H_B)\\
                    &\;\cdot\,\,e(\lambda \beta H_B, \lambda \gamma G_B)\\
                    &\;\cdot\,\,e(\lambda \beta H_B, (1 + \lambda \delta) H_B)\\
                    &=          1\\
                    &\;\cdot\,\,e(G_B, H_B)^{(1 + \lambda \alpha)(1 + \lambda \delta)}\\
                    &\;\cdot\,\,e(H_B, G_B)^{\lambda \beta \lambda \gamma}\\
                    &\;\cdot\,\,1\\
                    &=          e(G_B, H_B)^{(1 + \lambda \alpha)(1 + \lambda \delta) - \lambda^2 \beta \gamma}\\ 
                    &=     e(G_B, H_B),
\end{align*}

hence it is necessary that $(1 + \lambda \alpha)(1 + \lambda \delta) - \lambda^2 \beta \gamma = 1 \pmod {\ell_A^{e_A}}$, or equivalently 
$\lambda (\alpha + \delta) + \lambda^2 (\alpha \delta - \beta \gamma) = 0 \pmod {\ell_A^{e_A}}$. This must hold for \emph{any} choice of $\lambda$, in particular those that are invertible mod $\ell_A^{e_A}$, and hence it must hold that $\lambda (\alpha \delta - \beta \gamma) = -(\alpha + \delta) \pmod {\ell_A^{e_A}}$. Once more, this can only hold for \emph{any} $\lambda$ if $\alpha \delta - \beta \gamma = 0 \pmod {\ell_A^{e_A}}$ and $\alpha + \delta = 0 \pmod {\ell_A^{e_A}}$, or equivalently, $\delta = -\alpha \pmod {\ell_A^{e_A}}$ and $\alpha^2 + \beta\gamma = 0 \pmod {\ell_A^{e_A}}$. Therefore, in principle, such conditions should be obeyed to avoid Alice to find out Bob's choice. \qed

\subsubsection{Possible decryptions from a possible dishonest Bob}\label{sec:obs} $ $

Recalling $U, V \in E_B[\ell_A^{e_A}]$ are linearly independent points, and write $U = \alpha G_B + \beta H_B$, $V = \gamma G_B + \delta H_B$. Suppose Alice receives an information  $(E_B, \hat{G}_B, \\ \hat{H}_B)$ from Bob*. Then, Alice will compute  actually the degree-$\ell_A^{e_A}$ isogeny $\phi'_{A_{0}} : E_B \rightarrow E_{BA_{0}}$ whose kernel is $ker(\phi'_{A_{0}})$ =  $\langle G_B + r_A H_B \rangle$ and $\phi'_{A_{1}} : E_{B} \rightarrow E_{AB_{1}}$ whose kernel is $ker(\phi_{A_{1}})$ = $\langle (G_B + U) + r_A (H_B + V) \rangle$. It should be noted\footnote{See Theorem 9.6.18 from \cite{gal} and Proposition 12.12 from \cite{law}.} that if  $ker(\phi'_{A_{0}})$ $\subseteqq$  $ker(\phi_{A_{1}})$ then, $E_{BA_{0}}$ is isomorphic to $E_{BA_{1}}$ $i.e.$, $E_{BA_{0}} \cong E_{BA_{1}}$. Moreover, if $\phi_{A_{1}}$ is separable then there is a unique isogeny $\hat{\phi_{A}} : E_{BA_{0}} \rightarrow E_{BA_{1}}$. Now $(G_B + U) + r_A (H_B + V) = (G_B + \alpha G_B + \beta H_B) + r_A (H_B + \gamma G_B + \delta H_B) = (1 + \alpha + \gamma r_A) G_B + (r_A + \beta + \delta r_A) H_B$. Hence, by inspection, this point can only be in $\langle G_B + r_A H_B \rangle$  with the following conditions:

\begin{enumerate}

\item $(1 + \alpha + \gamma r_A)$ is invertible mod $\ell_A^{e_A}$ (i.e. if $\ell_A \nmid 1 + \alpha + \gamma r_A$); 

\item $(r_A + \beta + \delta r_A)/(1 + \alpha + \gamma r_A) = r_A \pmod {\ell_A^{e_A}}$, which means $\gamma r_A^2 + (\alpha + \delta)r_A - \beta = 0 \pmod {\ell_A^{e_A}}$ and hence $\gamma r_A^2 + (\alpha - \delta)r_A - \beta = 0 \pmod {\ell_A}$. Thus, a simple constraint on the coefficients ensures that the last equation has no solution then, just force $\ell_A \mid \gamma$ and $\ell_A \mid (\alpha - \delta)$, but $\ell_A \nmid \beta$.

\end{enumerate}

Therefore, it is important that this equation has no solution because, otherwise, if Alice and Bob* cannot control the coefficients $\alpha$, $\beta$, $\gamma$, $\delta$ apart from ensuring conditions as above, Bob* could be able to decrypt both messages from Alice. \qed 

\subsubsection{Wrapping up the conditions}\label{summing}$ $ 

In this section we will consider the three conditions on $\alpha$, $\beta$, $\gamma$, and $\delta$ based on the equations obtained in sections ~\ref{sec:indist} and ~\ref{sec:obs} to ensure $\binom{2}{1}$-SIOT protocol security in a scenario where Alice and Bob are dishonest  parties. Thus, conditions on $\alpha$, $\beta$, $\gamma$, and $\delta$ are obtained that guarantee that Alice will not be able to get the secret choice of Bob's bit $b$ and he will not be able to decipher both pairs $c_{0}$ and $c_{1}$ sent by Alice. Combining these relations yields $\gamma = -\alpha^2/\beta  \pmod {\ell_A^{e_A}}$ since $\beta$ is certainly invertible mod $\ell_A^{e_A}$. In particular, this means $V = -(\alpha/\beta)U$.

Additionally, in appendix~\ref{siot:sym} we will see the application of a symmetric pairing to analyze other possible conditions relative to the coefficients of points $U$ and $V$. Moreover, appendix~\ref{uv}  shows the process of sharing of these points.

\section{Conclusion of the security of the $\binom{2}{1}$-SIOT protocol}\label{conclusion:sec}

The security proof of the SIOT protocol is based on three parts, namely:
\begin{description}
\item[ i. ] The inherent security characteristics of~\cite{feo:jao}, that is, the computational problems mentioned in section~\ref{problem}; 

\item[ii. ] The privacy between a sender and receiver in a communication channel by means of Theorem~\ref{theorem};

\item[iii. ] An algebraic analysis that used \emph{Weil}'s pairing that defined some conditions necessary for a dishonest sender to does not cheat the security against an honest receive and \emph{vice versa}.
\end{description}

\section{Implementation of the $\binom{2}{1}$-SIOT protocol}\label{implementation:siot} 

The $\binom{2}{1}$-SIOT protocol was implement in the \emph{phyton} language, using a MacBook Air with a 1.6 GHz Intel Core i.5 processor, 4GB, 1.600MHz and DDR3 memory. Table~\ref{siot:table} shows the values of the prime numbers $p$ that were used in the implementation of the proposed protocol. 

\begin{table}[h]
   \centering
   
    \caption{Values used for $p$ in the $\binom{2}{1}$-SIOT protocol.}
    \begin{tabular}{ c  c  c  }
      \toprule
      $p = \ell^{e_{A}}_{A} \ell^{e_{B}}_{B} f \pm 1$ & Value & Size (bits)   \\
      \midrule
      $(3^4 \cdot 5^3 \cdot4) - 1$  & 40.499 & 16 \\
      $(3^6 \cdot 5^3 \cdot4) - 1$  & 364.499 & 19 \\
      $(3^7 \cdot 5^4 \cdot4) - 1$  & 5.467.499 & 23 \\
      $(3^{11} \cdot 5^6 \cdot4) - 1$  & 11.071.687.499 & 34  \\
      \bottomrule
    \end{tabular}
  
    \label{siot:table} 
    
  \end{table}
  
  In this work, it was only possible to use a maximum value of $p$ corresponding the size of 34 bits. Evidently, such $p$ values are insufficient to guarantee the security of the proposed protocol because~\cite{jao:reza} and~\cite{feo:jao} consider that the size of the $p- value$ for the security of a post-quantum cryptographic protocol based on isogenies of  elliptic curves is at least equal to 512 bits. However, this does not invalidate the proof of concept of  $\binom{2}{1}$-SIOT.
  
  \section{Performance estimate between some OT protocols}\label{perform:ot}

In table~\ref{tab:per} and figure~\ref{fig:estimar}, the number of types of operations by sender and receiver was verified in $\binom{2}{1}$-SIOT and both protocols SIDH-OT and WSW-OT from~\cite{vanessa:2019}. Thus, multiplications with scalar and point additions are denoted by \emph{Multi} and \emph{Add}, respectively. Furthermore, calculations for isogenies and pairing are denoted by \emph{Iso} and \emph{Png}, respectively. Therefore, we estimate that the $\binom{2}{1}$-SIOT protocol has slightly better performance than other two protocols. \\

\begin{table}[h]
\centering
\caption{Performance estimation.}
\begin{tabular}{|c|c|c|c|c|c|c|c|}
\hline
\multirow{2}{*}{Protocol}         & \multicolumn{3}{c|}{Sender}								    & \multicolumn{3}{c|}{Receiver}									  & \multirow{2}{*}{Png} \\ \cline{2-7}
                                               & \multicolumn{1}{l|}{Mult} & \multicolumn{1}{l|}{Add} & \multicolumn{1}{l|}{Iso} & \multicolumn{1}{l|}{Mult} & \multicolumn{1}{l|}{Add} & \multicolumn{1}{l|}{Iso} &					 \\ \hline
SIDH-OT				     & 8				    & 4				   & 4					    & 2					  & 2					  & 2				  & 2					 \\ \hline
WSW-OT				     & 3                                    & 3				   & 5			                     & 2					  & 1					  & 2				  & 4					 \\ \hline
SIOT				             & 3                                    & 5				   & 3					    & 2					  & 4					  & 2				  & -					 \\ \hline

\end{tabular}
\label{tab:per}
\end{table}

\begin{figure}[h]

\centering
\begin{tikzpicture}
\begin{axis}[
    ybar,
    enlargelimits=0.25,
    legend style={at={(0.5,-0.15)},
      anchor=north,legend columns=-1},
    ylabel={\#Total$\:$ Operations},
    symbolic x coords={SIOT,SIDH-OT,WSW-OT},
    xtick=data,
    nodes near coords,
    nodes near coords align={vertical},
    ]
\addplot coordinates {(SIOT,5) (SIDH-OT,10) (WSW-OT,5)};
\addplot coordinates {(SIOT,9) (SIDH-OT,6) (WSW-OT,4)};
\addplot coordinates {(SIOT,5) (SIDH-OT,6) (WSW-OT,7)};
\addplot coordinates {(SIOT,0) (SIDH-OT,2) (WSW-OT,4)};
\legend{Multiplications,Additions,Isogenies,Pairings}
\end{axis}

\end{tikzpicture}

\caption{Performance estimation.}

\label{fig:estimar}

\end{figure}

\section{Conclusion}\label{sec:conclusion}

In this paper, a proposal for a \emph{post quantum} protocol called SIOT is presented. Its security is based on the difficulty of an opponent to calculate isogenies between supersingular elliptic curves and the inspiration of the relative simplicity of the OT protocol of~\cite{cho:orl} to ensure privacy between the sender and the receiver. With respect to this privacy, it was important to elaborate a theorem, matching a privacy definition of~\cite{yl} with the computational problems of isogenies of~\cite{feo:jao}, considering a hypothetical scenario between a dishonest sender and an honest receiver and \emph{vice versa}. Finally, an algebraic analysis with \emph{Weil} pairing defined certain necessary conditions so that there were not  security and privacy violations in the proposed protocol. 

\appendix 

\section{Isogenies}\label{apend:iso}
\hspace{.5cm}In short, isogeny-based cryptography utilizes unique algebraic maps between elliptic curves that satisfy group homomorphism. This original idea introduced by \cite{bunov} detailed a \emph{Diffie-Hellman} cryptosystem based on the hardness of computing isogenies between ordinary elliptic curves. Nevertheless, \cite{chi} developed a \emph{quantum} algorithm that could compute isogenies between ordinary curves in subexponential time. This algorithm uses the fact that the structure of the elliptical group is commutative. Thus, ~\cite{feo:jao} adapted the isogeny-based key exchange protocol to be based on the difficulty of computing isogenies between supersingular elliptic curves, which does not have commutative endomorphism ring.

\begin{definition}

Let $E_{1}$ and $E_{2}$ be elliptic curves over $\mathbb{F}_{p}$. An isogeny over $\mathbb{F}_{p}$ is a morphism $\phi \colon E_{1} \to E_{2}$ over $\mathbb{F}_{p}$ such that  $\phi(\mathcal{\omicron}_{E_{1}})$ = $\mathcal{\omicron}_{E_{2}}$  is a group homomorphism. The zero isogeny is the constant map $\phi \colon E_{1} \to E_{2}$ given by $\phi(P) = \mathcal{\omicron}_{E_{2}}$ for all $P \in E(\mathbb{\bar{F}}_{p})$.
 If there is an isogeny between two elliptic curves $E_{1}$ and $E_{2}$ then:
 
 \begin{description}
 
\item[i.] $E_{1}$ and $E_{2}$ are isogenous;

\item[ii.] $\#E_{1}(\mathbb{F}_{p}) = \#E_{2}(\mathbb{F}_{p})$\footnote{See~\cite{tate}.};

\item[iii.]  $E_{1}$ and $E_{2}$ have the same $\jmath-invariant$ if and only if  $E_{1} \backsimeq  E_{2}$ over $\mathbb{\bar{F}}_{p}$ ($i.e.$ exists an isomorfism from $E_{1}$ to $E_{2}$)\footnote{See Theorem 9.3.6 from~\cite{gal}.}.

\end{description}

\end{definition}

\begin{definition}

Let $E_{1}$ and $E_{2}$ be elliptic curves over $\mathbb{F}_{p}$ and $\phi \colon E_{1} \to E_{2}$ over $\mathbb{F}_{p}$. The degree of a non-zero isogeny is the degree of the morphism. 
The degree of the zero isogeny is $0$. If there is an isogeny of degree $\ell$  between elliptic curves $E_{1}$ and $E_{2}$ then, they are $\ell$-isogenous.

\end{definition}

\begin{definition}\label{def:eleven}

Let $E_{1}$ and $E_{2}$ be elliptic curves over $\mathbb{F}_{p}$ and $\phi \colon E_{1} \to E_{2}$ an isogeny. Then, the kernel of an isogeny is $ker(\phi) = \{ P \in E_{1}(\bar{\mathbb{F}_{p}}) : \phi(P) = \mathcal{\omicron}_{E_{2}}\}$.

\begin{remark}

We can denote $E_{2} = E_{1}/ ker(\phi)$.

\end{remark}

\end{definition}

\begin{definition}\label{def:separable:kernel}

A non-zero isogeny separable $\phi \colon E_{1} \to E_{2}$  over $\mathbb{F}_{p}$ of $\ell$ - degree has  $\# ker(\phi) = \ell$. 

\end{definition}

\begin{definition}

Let $\phi \colon E_{1} \to E_{2}$  and $\hat{\phi} \colon E_{2} \to E_{3}$ be two isogenies with $\ell$-degree and $\hat{\ell}$-degree, respectively. Then, their composition is an isogeny $\hat{\phi}(\phi) \colon E_1 \to E_3$ with ($\ell  \cdot \hat{\ell}$)-degree.

\end{definition}

\begin{proposition}

Let $E_{1}$ be an elliptic curve over $\mathbb{F}_{p}$ and $\mathbb{G}$ a finite subgroup of $E_{1}(\mathbb{\bar{F}}_{p})$ that is defined over $\mathbb{F}_{p}$.  Then, there is an unique elliptic curve $E_{\mathbb{G}}$ and a separable isogeny $\phi \colon E_{1} \to E_{\mathbb{G}} = E_{1} / \mathbb{G}$ such that  $ker(\phi)= \mathbb{G}$.

\end{proposition}

\begin{proof}

See Theorem 25.1.6 and Corollary 25.1.7 from \cite{gal}. \qed

\end{proof}

\section{SIDH key exchange}~\label{apendice:sidh}
\hspace{.5cm}For a better understanding, see the first diagram in figure~\ref{fig:SIDH}. The sender and the receiver choose randomly two secret integers $r_A$ $\leftarrow$ $\ZZ/\ell_A^{e_A}\ZZ$  and $r_B$ $\leftarrow$ $\ZZ/\ell_B^{e_B}\ZZ$, respectively. Thus, the kernel $\langle P_A + r_{A}Q_{A}\rangle$ from sender has order $\ell_A^{e_A}$ and its secret key is computed as the degree $\ell_A^{e_A}$ isogeny $\phi_A :  E_{0} \rightarrow E_{A}$, and its $\pk_A$ is the isogenous curve $E_{A}$ together with images $G_{A} \leftarrow$ $\phi_{A}(P_{B})$ and $H_{A} \leftarrow$ $\phi_{A}(Q_{B})$. 

Similarly, the kernel $\langle P_B + r_{B}Q_{B}\rangle$ from receiver has order $\ell_B^{e_B}$ and its secret key is computed as the degree $\ell_B^{e_B}$ isogeny $\phi_B :  E_{0} \rightarrow E_{B}$, and its $\pk_B$ is the isogenous curve $E_{B}$ together with images $G_{B} \leftarrow$ $\phi_{B}(P_{A})$ and $H_{B} \leftarrow$ $\phi_{B}(Q_{A})$. In short, there is a public key exchange, say $\pk_A$ and $\pk_B$.  

To compute the shared secret $k$, sender uses its secret integers and receiver's public key to compute the degree $\ell_A$  isogeny $\phi^{'}_{A} : E_{B} \rightarrow E_{BA}$ whose kernel is the point $\phi_B(P_{A})$ + $r_{A}$$\phi_{B}(Q_{A})$ = $\phi_{B}(P_{A} + r_{A}Q_{A})$. In the same way, receiver uses its secret integers and sender's public key to compute the degree $\ell_B$  isogeny $\phi^{'}_{B} : E_{A} \rightarrow E_{AB}$ whose kernel is the point $\phi_A(P_{B})$ + $r_{B}$$\phi_{A}(Q_{B})$ = $\phi_{A}(P_{B} + r_{B}Q_{B})$.  It happens that $E_{BA}$ and $E_{AB}$ are isomorphic. Hence, both sender and receiver can compute a shared secret $k$, that is, there is the common $\jmath$-invariant $\jmath(E_{BA}) = \jmath(E_{AB})$. 

Therefore, $\mathcal{H}(\jmath(E_{BA}))$ = $\mathcal{H}(\jmath(E_{AB}))$ = $k$. After that, the sender computes $ c \leftarrow \mathcal{E}_{k}(m)$ and sends it to the receiver that computes $m \leftarrow \mathcal{D}_{k}(c)$. More details, see \cite{feo:jao}.

\begin{figure}[h]
 
 \resizebox{0.95\textwidth}{!}{%
 
 
\fbox{\pseudocode{
 \> \textbf{\emph{Supersingular Isogeny Diffie-Hellman} (SIDH)} \> \\ [0.1\baselineskip][\hline]
\textbf{Sender} \> \> \textbf{Receiver} \\
 \text{Input: } m  \in \mathcal{M}\> \> \text{Input: none}\\
 \text{Output: none}\> \> \text{Output: } m \\
 \text{$r_A \leftarrow \ZZ/\ell_A^{e_A}\ZZ$}\> \> \text{$r_B  \leftarrow \ZZ/\ell_B^{e_B}\ZZ$} \\
  \text{$\phi_A: E_{0} \rightarrow E_{A}$}\>  \> \text{$\phi_B: E_{0} \rightarrow E_{B}$} \\
 \text{$E_{A} \leftarrow E_{0} / \langle P_A + r_{A}Q_{A}\rangle$}\>  \> \text{$ E_{B} \leftarrow E_{0} / \langle P_B + r_{B}Q_{B}\rangle$} \\
 \text{$G_A \leftarrow \phi_A(P_B); H_A \leftarrow \phi_A(Q_B)$}   \> \> \text{$G_B \leftarrow \phi_B(P_A); H_B \leftarrow \phi_B(Q_A)$}  \\
 \pk_A \leftarrow (E_A, G_A, H_A) \>\> \pk_B \leftarrow (E_B, G_B, H_B) \\
  \text{}\>\sendmessageright{centercol=3,top= $\pk_A $} \> \text{$\phi^{'}_{B}: E_{A} \rightarrow E_{AB}$} \\
  \text{$\phi^{'}_{A}: E_{B} \rightarrow E_{BA}$}\> \sendmessageleft{centercol=3,top= $ \pk_B $} \>  \text{}\\
  \text{$E_{BA} \leftarrow E_{B}/ \langle G_{B} + r_{A}H_{B} \rangle$}\> \> \text{$E_{AB} \leftarrow E_{A}/ \langle G_{A} + r_{B}H_{A}\rangle$} \\
  \text{$ k  = \mathcal{H}( \jmath(E_{BA}))$ }\> \> \text{$k  = \mathcal{H}(\jmath(E_{AB}))$} \\
  \text{}\> \sendmessageright{centercol=3,top= $c \leftarrow \mathcal{E}_{k}(m) $} \>  \text{}\\
                     \> \> \text{$m$ = $\mathcal{D}_{k}(c)$} \\}
    } }
    
\caption{SIDH protocol.}
 \label{fig:SIDH}    
    
 \end{figure}
 
 \section{\emph{Oblivious Transfer} protocol}\label{ot:protocol}

 \hspace{.5cm}\emph{Oblivious Transfer} (OT) is a protocol in which a sender transfers one of many pieces of information to a receiver, but remains oblivious as to what piece has been transferred. The original notion of OT was first proposed by \emph{Michael Rabin} in 1981 \cite{rabin:1981} in which a sender sends an encrypted message to a receiver and this one could decrypt such message with probability 1/2. After this,  \cite{gold} presented a general form of OT, named 1-out-of-2 OT, $\binom{2}{1}$ - OT for short, $i.e$, where a sender sends two encrypted messages to a receiver being able to decrypt only one of them.
 
 Many authors have generalized this to $\binom{n}{1}$ - OT where the receiver chooses one message out of $n$ and $\binom{k}{n}$-OT in which the receiver chooses a subset of size $k$ from among $n$ messages. In this work, we will be focused only on $\binom{2}{1}$ - OT.
 
 \subsection{Protocol $\binom{2}{1}$- OT \emph{Chou-Orlandi}}\label{ot} 

\hspace{.5cm}In this appendix, we see the simplified scheme of the random OT proposed in~\cite{cho:orl}. 

\subsection*{Premises}

\begin{enumerate}
\item The scheme from~\cite{cho:orl} works in a primitive additive group $(\mathbb{G}, B, \mathbb{F}_{p}, +)$ of prime order $p$, generated by base point $B$;

\item  Let $s$ be a safety parameter and $\jmath \in \bin$. Thus, the hash function $\mathcal{H}: (\mathbb{G} \times \mathbb{G}) \times \mathbb{G} \rightarrow \{0,1\}^{s}$ is used to generate a cryptographic key $k_{\jmath}$ for use in a symmetric cipher defined by the functions $\mathcal{E}$ (encryption) and $\mathcal{D}$ (decryption), i.e, $c_0 = \mathcal{E}(k_{0}, m_{0})$ and $c_1 = \mathcal{E}(k_{1}, m_{1})$.   \\

\item Abstract view of information exchange from protocol $\binom{2}{1}$- OT \emph{Chou-Orlandi}.

\fbox{%
\centering
\pseudocode{%
  \< \< \\[-0.5\baselineskip]
  Sender \< \sendmessageright*{S} \< Receiver \\
  Sender \< \sendmessageleft*{R} \< Receiver \\ 
  Sender \< \sendmessageright*{(c_0, c_1)} \< Receiver \\
} 
}

\end{enumerate}

\subsection*{Sender and Receiver}

\subsubsection*{Setup - Sender}

\begin{enumerate}
\item Sender secretly chooses a value $y$ $\in$ $\mathbb{F}_{p}$;

\item Sender computes:

\begin{align*}
S &= yB \qquad \qquad   (1) \\
T &= yS;  \qquad \qquad (2) 
 \end{align*}

\item Sender sends  $S$  to  Receiver  which refuses if  $S \notin \mathbb{G}$.
\end{enumerate}

\subsubsection*{Setup - Receiver}

\begin{enumerate}
\item Receiver secretly chooses a value $x$  $\in$ $\mathbb{F}_{p}$;

\item Receiver computes: 

\begin{align*}
R &= b.S + x.B   \qquad (3), \: \: where \, b \in \bin \, is \,  chosen \,  by \,  Receiver; \\ 
\end{align*}

\item Receiver sends $R$ to Sender which refuses if  $R \notin \mathbb{G}$.

\end{enumerate}

\subsection*{Generation of cryptographic keys $k_{\jmath}$,  $\jmath \in \{ 0, 1\}$.}

\begin{enumerate}
\item Sender computes $k_{\jmath} = \mathcal{H}_{(S,R)} (yR - \jmath T)$;  \; \qquad \qquad (4)

\item Receiver computes $k_{b} = \mathcal{H}_{(S,R)} (bS + xB)$.  \qquad \qquad (5)
\end{enumerate}

\subsection*{Encryption and Decryption}

\begin{enumerate}
\item Sender  encrypts and sends $c = (c_0, c_1)$ to Receiver. Recalling $c_{0} = \mathcal{E}(k_{0}, M_{0})$ and  $c_{1}= \mathcal{E}(k_{1}, M_{1})$;

\item Receiver decrypts  and gets  $M_{b} = \mathcal{D}(k_{b} , c_{\jmath})$, $\jmath  \in \{0, 1\}$.
\end{enumerate}

\begin{remark}

 It is verified that a key $k_{\jmath}$, $j \in \{0,1\}$, is computed by  $\mathcal{H}_{(S, R)} [xyB + (b-j)T]$. Hence, at the end of the protocol if both parts are honest then we have that $k_{b} = k_{\jmath}$. In other words, if $\jmath = 0$ then $c = c_{0} = 0$ and $k_{0} = k_{b} = \mathcal{H}_{(S,R)}(xyB)$. Otherwise, if $\jmath = 1$ then $c = c_{1} = 1$ and $k_{1} = k_{b} = \mathcal{H}_{(S,R)}(xyB)$.

\begin{align*} 
k_\jmath &= yR - \jmath T; \\
      &= y(bS + xB) - \jmath T; \qquad \qquad from \, equation \: (3) \\
       &= byS + xyB - \jmath T; \\
       &= bT + xyB - \jmath T; \qquad \qquad from \, equations \: (1) \, and \, (2) \\ 
       &= xyB + (b-\jmath) T.
\end{align*}

\end{remark}

Therefore, we can conclude that if the Receiver chooses $ b \notin \jmath$, he will not share the secret (cryptographic key) with the Sender.

\section{ Linearly independent points}\label{linear:points}

\hspace{.5cm}In this appendix, we present definitions for the understanding of the process that determines the choice of linearly independent points $P_{A}, Q_{A}, P_{B}$ and $Q_{B}$ in the proposed protocol.

\begin{definition}[Frobenius] ~\label{frobenius}
Let $E(\mathbb{F}_q)$ be an ellipic curve, and let $E(\mathbb{F}_{q^k})$ be its $\mathbb{F}_{q^k}$-rational extension.
The \emph{Frobenius} map is the function $\Phi: E(\mathbb{F}_{q^k}) \rightarrow E(\mathbb{F}_{q^k})$ defined by $\Phi(x, y) = (x^q, y^q)$ for any $(x, y) \in E(\mathbb{F}_{q^k})$. $\Phi^i$ denotes its $i$-th self-composition, i.e. for any $P \in E(\mathbb{F}_{q^k})$, $\Phi^i(P) := P$ for $i = 0$, and $\Phi^i(P) = \Phi(\Phi^{i-1}(P))$ for $i > 0$.
\end{definition}

\begin{definition}[Trace]~\label{trace}
Let $E(\mathbb{F}_q)$ be an ellipic curve, and let $E(\mathbb{F}_{q^k})$ be its $\mathbb{F}_{q^k}$-rational extension.
The \emph{trace} map is the function $tr: E(\mathbb{F}_{q^k}) \rightarrow E(\mathbb{F}_{q^k})$ defined by $tr(P) = (1/k)\sum_{i=0}^{k-1}\Phi^i(P)$ where $1/k$ denotes the inverse of $k$ mod the order of $E(\mathbb{F}_{q^k})$. In particular, $k = 2$ for a supersingular curve in characteristic $p > 3$, and $tr(P) = (1/2)(P + \Phi(P))$.
\end{definition}

Hence, the trace map is important in that its eigenspaces, if nontrivial, form two linearly independent groups that can be used to sample points $P_A$, $Q_A$, $P_B$, $Q_B$ efficiently. Moreover, the trace definition assumes that $\gcd(k, \#E(\mathbb{F}_{q^k})) = 1$, which may not be the case, especially in the important setting where $\ell_A = 2$. Thus, for this scenario we also define the $quasi-trace$ map:

\begin{definition}[Quasi-trace]~\label{quasetrace}
Let $E(\mathbb{F}_q)$ be an ellipic curve, and let $E(\mathbb{F}_{q^k})$ be its $\mathbb{F}_{q^k}$-rational extension.
The \emph{quasi-trace} map is the function $tr: E(\mathbb{F}_{q^k}) \rightarrow E(\mathbb{F}_{q^k})$ defined by $tr(P) = \sum_{i=0}^{k-1}\Phi^i(P)$. In particular, $k = 2$ for a supersingular curve in characteristic $p > 3$, and $tr(P) = P + \Phi(P)$.
\end{definition}

\section{Possibility of symmetric pairings in the SIOT}\label{siot:sym}

\hspace{.5cm}Under certain circumstances, it is possible to define a \emph{symmetric} pairing $\hat{e}: E_B[\ell_A^{e_A}] \rightarrow F_{p^2}$. We now analyze the condition under which this can happen.
In what follows, recall that a distortion map is a linear transformation that maps a curve point to a linearly independent point.

The embedding degree for $E_B[\ell_A^{e_A}]$ is only 1, not 2 as it is for $E_0$, because $E_B$ is defined over $F_q$ with $q := p^2$, and since $p = (\ell_A^{e_A} \ell_B^{e_B} f)^2 - 1$, it follows that $\#E_B[\ell_A^{e_A}]  = (\ell_A^{e_A})^2 \mid q - 1$.
Hence a distortion map $\psi$ must map a point $P \in E[\ell_A^{e_A}](F_q)$ to a point $Q \in E[\ell_A^{e_A}](F_q)$ that is linearly independent from $P$, in which case $\psi$ linearly maps a basis $(G_B, H_B)$ to another basis $(G'_B, H'_B)$.

In particular, all coefficients of $\psi$ in basis $(G_B, H_B)$ must be integers mod $\ell_A^{e_A}$. For such a map to be a distortion map, it must have no eigenvectors (otherwise it would fail to map those points to linearly independent points), so we can simply require the characteristic polynomial to have no roots mod $\ell_A^{e_A}$.

In that case, the map $\psi(u G_B + v H_B) := v G_B - u H_B$ could be a suitable distortion map. Its characteristic polynomial is $\lambda^2 + 1$ which has no roots mod $\ell_A^{e_A}$ for a careful choice of $\ell_A$ (e.g. $\ell_A = 3$).
Now define the modified pairing $\hat{e}(P, Q) := e(P, \psi(Q))$ where $e(\cdot)$ is the \emph{Weil} pairing. Then:
\begin{align*}
\hat{e}(a G_B + b H_B, c G_B + d H_B) &= e(a G_B + b H_B, d G_B - c H_B)\\
&= e(a G_B, -c H_B) \cdot e(b H_B, d G_B)\\
&= e(G_B, H_B)^{-ac - bd},\\
\hat{e}(c G_B + d H_B, a G_B + b H_B) &= e(c G_B + d H_B, b G_B - a H_B)\\
&= e(c G_B, -a H_B) \cdot e(d H_B, b G_B)\\
&= e(G_B, H_B)^{-ac - bd},
\end{align*}
so this modified pairing is symmetric.

It remains to determine if it is isogeny-equivariant. If it is, a further constraint exists for the coefficients of $U$ and $V$, namely: 
\begin{align*}
\hat{e}(G_B + \lambda U, H_B + \lambda V) &= \hat{e}(G_B, H_B)^{(1 + \lambda \alpha)(1 + \lambda \delta)}\\
                    &\;\cdot\,\,\hat{e}(H_B, G_B)^{\lambda \beta \lambda \gamma}\\
                    &=          \hat{e}(G_B, H_B)^{(1 + \lambda \alpha)(1 + \lambda \delta) + \lambda^2 \beta \gamma}\\
                    &=     e(G_B, H_B),
\end{align*}
so we also need $(1 + \lambda \alpha)(1 + \lambda \delta) + \lambda^2 \beta \gamma = 1 \pmod {\ell_A^{e_A}}$. \qed 

\subsection{Taking symmetric pairings into account}$ $

Coupling the above constraints $\gamma = -\alpha^2/\beta  \pmod {\ell_A^{e_A}}$ and $\delta = -\alpha \pmod {\ell_A^{e_A}}$ with the additional condition $(1 + \lambda \alpha)(1 + \lambda \delta) + \lambda^2 \beta \gamma = 1 \pmod {\ell_A^{e_A}}$, we have
$(1 + \lambda \alpha)(1 - \lambda \alpha) - \lambda^2 \beta \alpha^2/\beta = 1 - 2\lambda^2 \alpha^2 = 1 \pmod {\ell_A^{e_A}}$ for any $\lambda$, or simply $2 \alpha^2 = 0\pmod {\ell_A^{e_A}}$, which has the solution $\alpha = \alpha_0 \cdot 2^{\lfloor e_A/2 \rfloor}$ for $\ell_A = 2$ and any $0 \leqslant \alpha_0 < 2^{\lceil e_A/2 \rceil}$, or $\alpha = \alpha_0 \cdot \ell_A^{\lceil e_A/2 \rceil}$ for $\ell_A \neq 2$ and any $0 \leqslant \alpha_0 < \ell_A^{\lfloor e_A/2 \rfloor}$. \qed

\section{Validating the process of sharing points (U, V)}\label{uv}

\hspace{.5cm}Now, we are going to verify the sharing of points $U$ and $V$ between Bob and Alice. This is important from the point of view of the correct functionality of the $\binom{2}{1}$-SIOT protocol with regard to the \emph{oblivious} characteristic, $i.e$, in practical terms, points $U$ and $V$  provide the sender to generate two secret keys. Thus, Bob defines points $U$ and $V$ as mentioned in section~\ref{def:premissas}. Recall that these points can be written as a linear combination. After that, he sends to Alice one of the pairs $(G_{B}, H_{B})$ or $(G_{B} - U, H_{B} - V)$. Obviously, Alice doesn't distinguish\footnote{See Subsection~\ref{sec:security}, Lemma~\ref{sec:lemma1}.} which pair of points she received. Thus, upon receipt of $\hat{G}_B$ and $\hat{H}_B$ points from Bob's public key, say $\hat{pk}_{B} = (E_B, \hat{G}_B, \hat{H}_B)$, Alice defines $\hat{U}$ and $\hat{V}$ yielding  $\hat{U} = U$ and $\hat{V} = V$. In other words, Alice and Bob have the assurance that points $U$ and $V$  are being correctly shared between the parties.

\begin{proof} $ $

\begin{description}

\item[1.] In a first assumption, Alice receives points $\hat{G}_{B} = (G_{B} - U)$ and $\hat{H}_{B} = (H_{B} - V)$  from Bob. Evidently, she has not  any knowledge about points  $G_{B}$ and $H_{B}$. Thus, she performs the algebraic development below.

\end{description}

\begin{align*}
\hat{U} &= \alpha \cdot \hat{G}_{B} +  \beta \cdot \hat{H}_{B} ; \\ 
\hat{U} &= \alpha \cdot (G_{B} - U) +  \beta \cdot (H_{B} - V) ; \\
\hat{U} &= \alpha \cdot G_{B} - \alpha \cdot U +  \beta \cdot H_{B} - \beta \cdot V ; \\
\hat{U} &= \underbrace{\alpha \cdot G_{B} +   \beta \cdot H_{B}}_{U} - (\alpha \cdot U +  \beta \cdot V) ;\\
\hat{U} &= U - (\alpha \cdot U +  \beta \cdot V) ; \\
\hat{U} &= U - \underbrace{[\alpha \cdot U +  \beta \cdot  (- \frac{\alpha}{\beta} \cdot U)]}_{0} ; \\ 
\hat{U} &= U. \\
%
\end{align*}

If \; $\hat{U} = U$  and $V = -(\alpha / \beta)U$, then  $\hat{V} = V$. \; \qed 

\begin{description}
            
\item[2.] In this second assumption, Alice receives $\hat{G}_{B} = G_{B}$ e $\hat{H}_{B} = H_{B}$ points from Bob. Similarly, 

\begin{align*}
\hat{U} &= \alpha \cdot \hat{G}_{B} +  \beta \cdot \hat{H}_{B} ; \\
\hat{U} &= \alpha \cdot G_{B} +  \beta \cdot H_{B} ; \\
\hat{U} &= \underbrace{\alpha \cdot G_{B} +   \beta \cdot H_{B}}_{U} ;\\
\hat{U} &= U. \\
\end{align*}

Evidently, in this case, If \; $\hat{U} = U$ then, $\hat{V} = V$ \; \qed.

\end{description}
\end{proof}

\end{document}